\documentclass[11pt]{article}
\usepackage{authblk}
\usepackage{amsthm}
\usepackage{graphicx}
\usepackage{amsmath}
\usepackage{amssymb}
\usepackage{eucal}
\usepackage{url}
\usepackage{amsfonts}
\usepackage{pgf,tikz}
\usepackage{pgfplots}
\usepackage{cmap}
\usepackage{fancyhdr}
\usepackage{float}
\usetikzlibrary{arrows,automata,positioning,decorations.pathreplacing,decorations.markings}
\usepackage{cite,tabularx,color}
\usepackage{caption}
\usepackage{tcolorbox}
\usepackage{gastex}

\usepackage{mathptmx}

\newcommand{\dt}{.}

\newtheorem{thm}{Theorem}

\newtheorem{pro}[thm]{Proposition}
\newtheorem{cor}[thm]{Corollary}
\newtheorem{lem}[thm]{Lemma}

\newcommand{\sa}{synchronizing automata}

\newcommand{\sw}{synchronizing word}
\newcommand{\csw}{carefully synchronizing word}
\newcommand{\csws}{carefully synchronizing words}
\newcommand{\sws}{synchronizing words}

\newcommand{\esw}{exactly synchronizing word}
\newcommand{\esws}{exactly synchronizing words}
\newcommand{\scn}{strongly connected}

\DeclareSymbolFont{rsfscript}{OMS}{rsfs}{m}{n}
\DeclareSymbolFontAlphabet{\mathrsfs}{rsfscript}

\newcommand{\mA}{\mathrsfs{A}}
\newcommand{\mB}{\mathrsfs{B}}
\newcommand{\mC}{\mathrsfs{C}}
\newcommand{\mE}{\mathrsfs{E}}
\newcommand{\mH}{\mathrsfs{H}}
\newcommand{\mP}{\mathrsfs{P}}
\newcommand{\mW}{\mathrsfs{W}}
\newcommand{\Cerny}{\v{C}ern\'{y}}

\title{Careful synchronization\\ of partial deterministic finite automata}

\author[1,2]{Hanan Shabana\thanks{hananshabana22@gmail.com; supported by the Competitiveness Enhancement Program of Ural Federal University.}}
\author[2]{M. V. Volkov\thanks{m.v.volkov@urfu.ru; supported by Ural Mathematical Center under agreement No. 075-02-2020-1537/1 with the Ministry of Science and Higher Education of the Russian Federation.}}
\affil[1]{Faculty of Electronic Engineering, Menoufia University, Menouf, Egypt}
\affil[2]{Institute of Natural Sciences and Mathematics, Ural Federal University, Ekaterinburg, Russia}
\date{}

\begin{document}

\maketitle

\begin{abstract}
We approach the task of computing a carefully synchronizing word of minimum length for a given partial deterministic automaton, encoding the problem as an instance of SAT and invoking a SAT solver. Our experiments demonstrate that this approach gives satisfactory results for automata with up to 100 states even if very modest computational resources are used. We compare our results with the ones obtained by the first author for exact synchronization, which is another version of synchronization studied in the literature, and draw some theoretical conclusions.\\[.2ex]

\end{abstract}

\section{Background and motivation}
\label{sec:intro}
\setcounter{footnote}{0}

A \emph{deterministic finite automaton} is a triple $\langle Q,\Sigma,\delta\rangle$, where $Q$ and $\Sigma$ are finite sets called the \emph{state set} and the \emph{input alphabet}, respectively, and $\delta\colon Q\times\Sigma\to Q$ is a (not necessarily total) map. The elements of $Q$ and $\Sigma$ are called \emph{states} and \emph{letters}, respectively, and $\delta$ is referred to as the \emph{transition function}. If $\mA=\langle Q,\Sigma,\delta\rangle$ is such that the function $\delta$ is totally defined, we say that $\mA$ is a \emph{complete deterministic finite automaton} (CFA); otherwise, $\mA$ is called a \emph{partial deterministic finite automaton} (PFA).

Let $\Sigma^*$ stand for the set of all words over $\Sigma$, including the empty word, denoted $\varepsilon$, and let $\mathcal{P}(Q)$ be the power set of $Q$. The transition function $\delta$ extends to a function $\mathcal{P}(Q)\times\Sigma^*\to\mathcal{P}(Q)$, still denoted $\delta$, in the following inductive way: for every subset $S\subseteq Q$ and every word $w\in\Sigma^*$, we set
\[
\delta(S,w):=\begin{cases}S &\text{ if } w=\varepsilon, \\
\{\delta(q,a)\mid q\in\delta(S,v)\}&\text{ if $w=va$ with $v\in\Sigma^*$ and $a\in\Sigma$.}
\end{cases}
\]
(The set $\delta(S,v)$ in the right-hand side is defined by the inductive assumption.) Observe that $\delta(S,w)$ may be empty. If $\delta(S,w)$ is empty and $S$ consists of a single state $q$, we say that $w$ is \emph{undefined at $q$}; otherwise $w$ is said to be \emph{defined at $q$}.

When dealing with a fixed automaton, we write $q\dt w$ for $\delta(q,w)$ and $S\dt w$ for $\delta(S,w)$. We also simplify the notation for $\mA=\langle Q,\Sigma,\delta\rangle$, writing $\mA=\langle Q,\Sigma\rangle$.

A CFA $\mA=\langle Q,\Sigma\rangle$ is called \emph{synchronizing} if it possesses a word $w\in\Sigma^*$ whose action leaves the automaton in one particular state no matter at which state in $Q$ it is applied: $q\dt w=q'\dt w$ for all $q,q'\in Q$. Any $w$ with this property is said to be a \emph{synchronizing} word for the automaton. We refer the reader to the survey~\cite{Vo08} and the chapter~\cite{KV} of the forthcoming `Handbook of Automata Theory' for a discussion of the rich theory of complete \sa\ as well as their diverse connections and applications.

In the literature, there are two widely studied extensions of the concept of synchronization to PFAs: \emph{careful synchronization} and \emph{exact synchronization}. Careful synchronization was introduced by Rystsov~\cite{Rystsov:1980} and studied in detail by Martyugin~\cite{Martyugin08,Martyugin10,Martyugin12,Martyugin13,Martyugin14}\footnote{It should be mentioned that Rystsov~\cite{Rystsov:1980} used the term `synchronizing word' for what we call `carefully synchronizing word', following Martyugin~\cite{Martyugin08,Martyugin10,Martyugin12,Martyugin13,Martyugin14}. The authors are grateful to Dr. Pavel Panteleev who drew their attention to Rystsov's paper.}. A PFA $\mA=\langle Q,\Sigma\rangle$ is said to be \emph{carefully synchronizing} if there is a word $w=a_1\cdots a_\ell$, with $a_1,\dots,a_\ell\in\Sigma$, that satisfies the following conditions:
\begin{enumerate}
 	\itemindent=16pt \itemsep=-4pt
 	\item[$(C1)$:] the letter $a_1$ is defined at every state in $Q$;
 	\item[$(C2)$:] the letter $a_t$ with $1<t\le \ell$ is defined at every state in $Q.a_1\cdots a_{t-1}$,
 	\item[$(C3)$:] $|Q.w|=1$.
\end{enumerate}
Any word $w$ satisfying $(C1)$--$(C3)$ is called a \emph{carefully synchronizing word} for $\mA$. Thus, when a carefully synchronizing word is applied at any state in $Q$, no undefined transition occurs during the course of application. The PFA $\mP_4$ in Fig.~\ref{fig:example} serves as an example of a carefully synchronizing PFA; one can check that the word $a^2baba^2$ maps each of its states to the state 2 and bypasses the only undefined transition of $\mP_4$.
\begin{figure}[h]
\unitlength=1mm
\begin{picture}(20,24)(-49,-5)
\node(A)(0,20){1}
\node(B)(20,20){2}
\node(C)(20,0){3}
\node(D)(0,0){4}
\drawedge[curvedepth=3](A,B){$a$}
\drawedge[curvedepth=-3](A,B){$b$}
\drawedge(B,C){$b$}
\drawedge(C,D){$b$}
\drawedge(D,A){$a$}
\drawloop[loopangle=0](B){$a$}
\drawloop[loopangle=-0](C){$a$}
\end{picture}
\caption{The automaton $\mP_4$}\label{fig:example}
\end{figure}
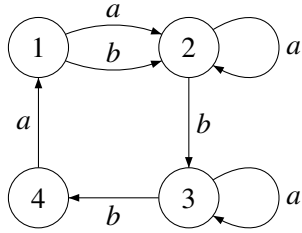

If a word $w$ satisfies the condition $(C3)$, it is called an \emph{exactly synchronizing word} for $\mA$. Thus, $w$ can be undefined at some states in $Q$ but there must be a state at which $w$ is defined and $q\dt w=q'\dt w$ for all $q,q'\in Q$ at which $w$ is defined. Clearly, a carefully synchronizing word is exactly synchronizing but the converse needs not be true. A PFA is said to be \emph{exactly synchronizing} if it possesses an exactly synchronizing word. The class of exactly synchronizing PFAs is much larger than that of carefully synchronizing PFAs: for instance, if one adds to a synchronizing CFA a new state at which no letter is defined, one gets an exactly synchronizing PFA which is not carefully synchronizing since the condition $(C1)$ fails for each letter. Observe that when restricted to CFAs, both careful synchronization and exact synchronization coincide with the above `standard' notion of synchronization.

Both versions of synchronization that we have defined have interesting connections and numerous applications.

Careful synchronization is relevant in industrial robotics where \sa\ are widely used to design feeders, sorters, and orienters that work with flows of certain objects carried by a conveyer. The goal is achieved by making the flow encounter passive obstacles placed appropriately along the conveyer belt. The situation can be modelled by an finite automaton whose states represent possible orientations of the objects while the action of letters represents the effect of obstacles. Then a \sw\ corresponds to a sequence of obstacles that forces the objects take a prescibed orientation. We refer to Natarajan~\cite{Na86,Na89} for the origin of this automata-based approach; a transparent illustrative example can be found in~\cite{AnVo04}. In practice, objects to be oriented or sorted often have fragile parts that could be damaged if hitting an obstacle. In order to prevent any damage, we have to forbid `dangerous' transitions in the automaton modelling the orienter/sorter so that the automaton becomes partial, with \csws\ corresponding to `safe' obstacle sequences. (Actually, the term `careful synchronization' has been selected with this application in mind.) Another application comes from the coding theory\footnote{We refer the reader to \cite[Chapters~3 and~10]{Berstel&Perrin&Reutenauer:2009} for a detailed account of profound connections between codes and automata.}. Recall that a \emph{prefix code} over a finite alphabet $\Sigma$ is a set $X\subset\Sigma^*$ such that no word of $X$ is a prefix of another word of $X$. Decoding of a finite prefix code $X$ over $\Sigma$ can be implemented by a finite deterministic automaton $\mA_X$ whose state $Q$ is the set of all proper prefixes of the words in $X$ (including the empty word $\varepsilon$) and whose transitions are defined as follows: for $q\in Q$ and $a\in\Sigma$,
\begin{displaymath}
 q\dt a =\begin{cases} qa & \text{if $qa$ is a proper prefix of a word of $X$},\\
 \varepsilon & \text{if $qa \in X$},\\
 \text{undefined} & \text{otherwise}.
\end{cases}
\end{displaymath}
In general, $\mA_X$  is a PFA (it is complete if and only if the code $X$ is not contained in another prefix code over $\Sigma$). It can be shown that if $\mA_X$ is carefully synchronizing, the code $X$ enjoys a very useful property: should a channel error occur, it suffices to transmit a \csw\ $w$ of $\mA_X$ such that $Q\dt w=\{\varepsilon\}$ to ensure that the next symbols will be decoded correctly.

Exact synchronization is relevant in biologically inspired computing where exactly synchronizing words appear under the name `constants' in the study of so-called splicing systems, see Bonizzoni and Jonoska~\cite{splicing}. Another cause of interest in exact synchronization is provided by so-called $\epsilon$-machines, important models in the theory of stationary information sources, see Travers and Crutchfield \cite{travers,travers-1}, where the term `exact synchronization' comes from.

In view of the connections just outlined, the problems of determining whether or not a given PFA is carefully or exactly synchronizing and of finding its shortest, carefully or exactly \sws\ are both natural and important. The bad news is that these problems turn out to be quite difficult. For careful synchronization, deciding whether a given PFA is carefully synchronizing is known to be PSPACE-complete and the minimum length of \csws\ for carefully synchronizing PFAs can be exponential as a function of the number of states. (These results were found by Rystsov~\cite{Rystsov:1980,Rystsov:1983} and later rediscovered and strengthened by Martyugin~\cite{Martyugin12}.) For exact synchronization, the situation is better in the strongly connected case. (Recall that a PFA $\mA=\langle Q,\Sigma\rangle$ is said to be \emph{strongly connected} if for every pair $(q,q')\in Q\times Q$, there exists a word $w\in\Sigma^*$ such that $q'=q\dt w$.) Namely, checking whether a given \scn\ PFA is exactly synchronizing can be done in polynomial time, and for \scn\ exactly synchronizing PFAs with $n$ states, there exits a cubic in $n$ upper bound on the minimum length of \esws---both these facts readily follow from a result in~\cite{travers}. However, Berlinkov~\cite{Berlinkov14} has shown that in the absence of strong connectivity, testing a given PFA for exact synchronization becomes PSPACE-complete; he has also constructed is a series of $n$-state PFAs whose shortest \esws\ have length of magnitude $2^{\Omega(n)}$. Thus, for the general case, problems related to exact synchronization are no less complicated than those related to careful synchronization.

We conclude that to attack synchronization issues in the realm of PFAs, one inevitably has to employ approaches that have proved to be efficient for dealing with computationally hard problems. One of popular such approaches consists in encoding instances of a problem of interest into instances of the Boolean satisfiability problem (SAT) that are then fed to a SAT solver, i.e., a specialized program designed to solve instances of SAT. We refer the reader to the survey~\cite{HKR} or to the handbook~\cite{BHvMW} for a detailed discussion of the SAT-solver approach and impressive examples of its successful applications in various areas.

For studying synchronization of CFAs, the SAT-solver approach was initiated by Skvortsov and Tipikin \cite{Skvortsov:2011} and G\"uni\c{c}en, Erdem, and Yenig\"un \cite{Gunicen:2013}. The present authors extended the approach to careful synchronization of PFAs in the conference paper~\cite{Motor}, while the paper~\cite{Shiop} by the first author dealt with exact synchronization. We mention also our earlier articles \cite{Sh18,ShVo18}, where the SAT-solver approach was applied for a study of synchronizing nondeterministic automata.

The present paper is an extended and augmented version of~\cite{Motor}. On the theoretical side, we include several new results  and provide details and proofs omitted in~\cite{Motor} due to space constraints. Besides that, we exhibit many additional experimental results, including a comparative study of careful and exact synchronization for certain families of PFAs.

\section{Reduction to SAT}
\label{sec:SAT}

For completeness, recall the formulation of the Boolean satisfiability problem (SAT). An instance of SAT is a pair $(V,C)$, where $V$ is a set of Boolean variables and $C$ is a collection of clauses over $V$. (A \emph{clause} over $V$ is a disjunction of literals and a \emph{literal} is either a variable in $V$ or the negation of a variable in~$V$.) Any \emph{truth assignment} on $V$, i.e., any map $\varphi\colon V\to\{0,1\}$, extends to a map $C\to\{0,1\}$ (still denoted by $\varphi$) via the usual rules of propositional calculus: $\varphi(\neg x)=1-\varphi(x)$, $\varphi(x\vee y)=\max\{\varphi(x),\varphi(y)\}$. A truth assignment $\varphi$ \emph{satisfies} $C$ if $\varphi(c)=1$ for all $c\in C$. The answer to an instance $(V,C)$ is YES if $(V,C)$ has a \emph{satisfying assignment} (i.e., a truth assignment on $V$ that satisfies $C$) and NO otherwise.

We aim to effectively reduce the following problem to SAT.

 \smallskip

\begin{tcolorbox}
 \noindent CSW (the existence of a \csw\ of a given length):

 \noindent\textsc{Input}: a PFA $\mathrsfs{A}$ and a positive integer $\ell$ (given in unary);

 \noindent\textsc{Output}: YES if $\mathrsfs{A}$ has a \csw\ of length $\ell$;\\\phantom{\textsc{Output}:} NO otherwise.
\end{tcolorbox}

 \smallskip

\textbf{Remark 1.} We have to assume that the integer $\ell$ is given in unary because with $\ell$ given in binary, a polynomial time reduction from CSW to SAT is hardly possible. Indeed, it easily follows from~\cite{Martyugin12} that the version of CSW in which the integer parameter is given in binary is PSPACE-hard, and the existence of a polynomial reduction from a PSPACE-hard problem to SAT would imply that the polynomial hierarchy collapses at level~1. In contrast, the version of CSW with the unary integer parameter is easily seen to belong to NP: given an instance $(\mathrsfs{A}=\langle Q,\Sigma\rangle,\ell)$ of CSW in this setting, guessing a word $w\in\Sigma^*$ of length $\ell$ is legitimate. Then one just checks whether or not $w$ is carefully synchronizing for $\mathrsfs{A}$, and time spent for this check is clearly polynomial in the size of $(\mathrsfs{A},\ell)$.

 \smallskip

Now, given an arbitrary instance $(\mathrsfs{A},\ell)$ of CSW, we construct an instance $(V,C)$ of SAT such that the answer to $(\mathrsfs{A},\ell)$ is YES if and only if so is the answer to $(V,C)$. Even though our encoding follows general patterns presented in \cite[Chapters~2 and~16]{BHvMW}, it has some specific features so that we describe it in full detail and provide a rigorous proof of its adequacy. In the following presentation of the encoding, precise definitions and statements are interwoven with less formal comments explaining the `physical' meaning of variables and clauses.

 So, take a PFA $\mathrsfs{A}=\langle Q,\Sigma\rangle$ and an integer $\ell>0$. Denote the sizes of $Q$ and $\Sigma$ by $n$ and $m$ respectively, and fix some numbering of these sets so that $Q=\{q_1,\dots,q_n\}$ and $\Sigma=\{a_1,\dots,a_m\}$.

 We start with introducing the variables used in the instance $(V,C)$ of SAT that encodes $(\mathrsfs{A},\ell)$. The set $V$ consists of two sorts of variables: $m\ell$ \emph{letter variables} $x_{i,t}$ with $1\le i\le m$, $1\le t\le\ell$, and $n(\ell+1)$ \emph{state variables} $y_{j,t}$ with $1\le j\le n$, $0\le t\le\ell$. We use the letter variables to encode the letters of a hypothetical \csw\ $w$ of length $\ell$: namely, we want the value of the variable $x_{i,t}$ to be 1 if and only if the $t$-th letter of $w$ is $a_i$. The intended meaning of the state variables is as follows: we want the value of the variable $y_{j,t}$ to be 1 whenever the state $q_j$ belongs to the image of $Q$ under the action of the prefix of $w$ of length $t$, in which situation we say that $q_j$ \emph{is active after $t$ steps}. We see that the total number of variables in $V$ is $m\ell+n(\ell+1)=(m+n)\ell+n$.

 Now we turn to constructing the set of clauses $C$. It consists of four groups. The group $I$ of \emph{initial clauses} contains $n$ one-literal clauses $y_{j,0}$, $1\le j\le n$, and expresses the fact that all states are active after 0 steps.

 For each $t=1,\dots,\ell$, the group $L$ of \emph{letter clauses} includes the clauses
 \begin{equation}
 \label{eq:only one letter}
 x_{1,t}\vee\dots\vee x_{m,t},\quad \neg x_{r,t}\vee\neg x_{s,t},\ \text{ where }\ 1\le r<s\le m.
 \end{equation}
 Clearly, the clauses \eqref{eq:only one letter} express the fact that the $t$-th position of our hypothetical \csw\ $w$ is occupied by exactly one letter in $\Sigma$. Altogether, $L$ contains $\ell\left(\frac{m(m-1)}2+1\right)$ clauses.

 For each $t=1,\dots,\ell$ and each triple $(q_j,a_i,q_k)$ such that $q_j\dt a_i=q_k$, the group $T$ of \emph{transition clauses} includes the clause
 \begin{equation}
 \label{eq:transition}
 \neg y_{j,t-1}\vee\neg x_{i,t}\vee y_{k,t}.
 \end{equation}
 Invoking the basic laws of propositional logic, one sees that the clause \eqref{eq:transition} is equivalent to the implication
 $y_{j,t-1}\mathop{\&} x_{i,t}\to y_{k,t},$
 that is, \eqref{eq:transition} expresses the fact that if the state $q_j$ has been active after $t-1$ steps and $a_i$ is the $t$-th letter of $w$, then the state $q_k=q_j\dt a_i$ becomes active after $t$ steps. Further, for each $t=1,\dots,\ell$ and each pair $(q_j,a_i)$ such that $a_i$ is undefined at $q_j$ in $\mA$, we add to $T$ the clause
 \begin{equation}
 \label{eq:undefined}
 \neg y_{j,t-1}\vee\neg x_{i,t}.
 \end{equation}
The clause is equivalent to the implication
 $y_{j,t-1}\to\neg x_{i,t},$
and thus, it expresses the requirement that the letter $a_i$ should not occur in the $t$-th position of $w$ if $q_j$ has been active after $t-1$ steps. Obviously, this corresponds to the conditions $(C1)$ (for $t=0$) and $(C2)$ (for $t>0$) in the definition of careful synchronization. For each $t=1,\dots,\ell$ and each pair $(q_j,a_i)\in Q\times\Sigma$, exactly one of the clauses \eqref{eq:transition} or \eqref{eq:undefined} occurs in $T$, whence $T$ consists of $\ell mn$ clauses.

The final group $S$ of  \emph{synchronization clauses} includes the clauses
 \begin{equation}
 \label{eq:only one state}
 \neg y_{r,\ell}\vee\neg y_{s,\ell},\ \text{ where }\ 1\le r<s\le n.
 \end{equation}
The clauses \eqref{eq:only one state} just say that at most one state remains active when the action of the word $w$ is completed, which corresponds to the condition $(C3)$ from the definition of careful synchronization. The group $S$ contains $\frac{n(n-1)}2$ clauses.

Summing up, the number of clauses in $C:=I\cup L\cup T\cup S$ is
 \begin{equation}
 \label{eq:number}
 n+\ell\left(\tfrac{m(m-1)}2+1\right)+\ell mn+\tfrac{n(n-1)}2=
 \ell\left(\tfrac{m(m-1)}2+mn+1\right)+\tfrac{n(n+1)}2.
 \end{equation}
In comparison with encodings used in our earlier papers~\cite{ShVo18,Sh18}, the encoding suggested here produces much smaller SAT instances\footnote{It is fair to say that encodings in~\cite{ShVo18,Sh18} were designed to handle  much more general nondeterministic automata so it is not a surprise that those encodings were bulkier than the present one.}. Since in the applications the size of the input alphabet is a (usually small) constant, the leading term in \eqref{eq:number} is $\Theta(\ell n)$ while the restrictions to PFAs of the encodings from~\cite{ShVo18,Sh18} have $\Theta(\ell n^2)$ clauses. We discuss at the end of this section (see Remarks~3 and~4) how one can further reduce the number of clauses involved.

\begin{thm}
\label{thm:reduction}
A PFA $\mathrsfs{A}$ has a \csw\ of length $\ell$ if and only if the instance $(V,C)$ of SAT constructed above is satisfiable. Moreover, the \csws\ of length $\ell$ for $\mathrsfs{A}$ are in a 1-1 correspondence with the restrictions of satisfying assignments of $(V,C)$ to the letter variables.
\end{thm}

 \begin{proof}
 	Suppose that $\mathrsfs{A}$ has a \csw\ of length $\ell$. We fix such a word $w$ and denote by $w_t$ its prefix of length $t=1,\dots,\ell$. Define a truth assignment $\varphi\colon V\to\{0,1\}$ as follows: for $1\le i\le m$, $0\le j\le n$, $1\le t\le\ell$, let
 	\begin{align}
 	\label{eq:letter variables}
 	\varphi(x_{i,t})&:=\begin{cases}
 	1 &\text{if the $t$-th letter of $w$ is $a_i$,}\\
 	0 &\text{otherwise;}
 	\end{cases}\\
 	\label{eq:initial variables}
 	\varphi(y_{j,0})&:=1;\\
 	\label{eq:state variables}
 	\varphi(y_{j,t})&:=\begin{cases}
 	1 &\text{if the state $q_j$ lies in $Q\dt w_t$,}\\
 	0 &\text{otherwise.}
 	\end{cases}
 	\end{align}
 	In view of \eqref{eq:letter variables} and \eqref{eq:initial variables}, $\varphi$ satisfies all clauses in $L$ and respectively $I$. As $w_\ell=w$ and $|Q\dt w|=1$, we see that \eqref{eq:state variables} ensures that $\varphi$ satisfies all clauses in $S$. It remains to analyze the clauses in $T$. For each fixed $t=1,\dots,\ell$, these clauses are in a 1-1 correspondence with the pairs in $Q\times\Sigma$. We fix such a pair $(q_j,a_i)$, denote the clause corresponding to  $(q_j,a_i)$ by $c$ and consider three cases.
 	
\smallskip

\noindent\emph{\textbf{Case 1}: the letter $a_i$ is not the $t$-th letter of $w$}. Here $\varphi(x_{i,t})=0$ by \eqref{eq:letter variables}, and hence, $\varphi(c)=1$ as $\neg x_{i,t}$ occurs in $c$, independently of $c$ having the form \eqref{eq:transition} or \eqref{eq:undefined}.

\smallskip
 	
\noindent\emph{\textbf{Case 2}: the letter $a_i$ is the $t$-th letter of $w$ but it is undefined at $q_j$}. Here the clause $c$ must be of the form \eqref{eq:undefined}. Observe that $t>1$ in this case since the first letter of the \csw\ $w$ must be defined at each state in $Q$. Moreover, the state $q_j$ cannot belong to the set $Q\dt w_{t-1}$ because $a_i$ must be defined at each state in this set. Hence $\varphi(y_{j,t-1})=0$ by \eqref{eq:state variables}, and $\varphi(c)=1$ as the literal $\neg y_{j,t-1}$ occurs in $c$.

\smallskip
 	
\noindent\emph{\textbf{Case 3}: the letter $a_i$ is the $t$-th letter of $w$ and it is defined at $q_j$}. Here the clause $c$ must be of the form \eqref{eq:transition}, in which the literal $y_{k,t}$ corresponds to the state $q_k=q_j\dt a_i$. If the state $q_j$ does not belong to the set $Q\dt w_{t-1}$, then as in the previous case, we have $\varphi(y_{j,t-1})=0$ and $\varphi(c)=1$. If $q_j$ belongs to $Q\dt w_{t-1}$, then the state $q_k$ belongs to the set $(Q\dt w_{t-1})\dt a_i=Q\dt w_t$, whence $\varphi(y_{k,t})=1$ by \eqref{eq:state variables}. We conclude that $\varphi(c)=1$ as the literal $y_{k,t}$ occurs in $c$.
 	
 	\smallskip
 	
 	Conversely, suppose that $\varphi\colon V\to\{0,1\}$ is a satisfying assignment for $(V,C)$. Since $\varphi$ satisfies the clauses in $L$, for each $t=1,\dots,\ell$, there exists a unique $i\in\{1,\dots,m\}$ such that $\varphi(x_{i,t})=1$. This defines a map $\chi\colon\{1,\dots,\ell\}\to\{1,\dots,m\}$. Let $w:=a_{\chi(1)}\cdots a_{\chi(\ell)}$. We aim to show that $w$ is a \csw\ for $\mA$, i.e., that $w$ fulfils the conditions $(C1)$--$(C3)$. For this, we first prove two auxiliary claims. Recall that a state is said to be active after $t$ steps if it lies in $Q\dt w_t$, where, as above, $w_t$ is the length $t$ prefix of the word $w$. (By the length 0 prefix we understand the empty word $\varepsilon$.)
 	
 	\smallskip
 	
 	\emph{\textbf{Claim 1}. For each $t=0,1,\dots,\ell$, there are states active after $t$ steps}.
 	
 	\emph{\textbf{Claim 2}. If a state $q_k$ is active after $t$ steps, then $\varphi(y_{k,t})=1$}.
 	
 	We prove both claims simultaneously by induction on $t$. The induction basis $t=0$ is guaranteed by the fact that all states are active after 0 steps and $\varphi$ satisfies the clauses in $I$. Now suppose that $t>0$ and there are states active after $t-1$ steps. Let $q_r$ be such a state. Then $\varphi(y_{r,t-1})=1$ by the induction assumption. Let $i:=\chi(t)$, that is, $a_i$ is the $t$-th letter of the word $w$. Then $\varphi(x_{i,t})=1$, whence $\varphi$ cannot satisfy the clause of the form \eqref{eq:undefined} with $j=r$. Hence this clause cannot appear in $T$ as $\varphi$ satisfies the clauses in $T$. This means that the letter $a_i$ is defined at $q_r$ in $\mA$, and the state $q_s:=q_r\dt a_i$ is active after $t$ steps. Claim 1 is proved.
 	
 	Now let $q_k$ be an arbitrary state that is active after $t>0$ steps. Since $a_i$ is the $t$-th letter of $w$, we have $Q\dt w_t=(Q\dt w_{t-1})\dt a_i$, whence
 	$q_k=q_j\dt a_i$ for same $q_j\in Q\dt w_{t-1}$. Therefore the clause \eqref{eq:transition} occurs in $T$, and thus, it is satisfied by $\varphi$. Since $q_j$ is active after $t-1$ steps, $\varphi(y_{j,t-1})=1$ by the induction assumption; besides that, $\varphi(x_{i,t})=1$. We conclude that in order to satisfy \eqref{eq:transition}, the assignment $\varphi$ must fulfil $\varphi(y_{k,t})=1$. This completes the proof of Claim~2.
 	
 	\smallskip
 	
 	We turn to prove that the word $w$ fulfils $(C1)$ and $(C2)$. This amounts to verifying that for each $t=1,\dots,\ell$, the $t$-th letter of the word $w$ is defined at every state $q_j$ that is active after $t-1$ steps. Let, as above, $a_i$ stand for the $t$-th letter of $w$. If $a_j$ were undefined at $q_j$, then by the definition of the set $T$ of transition clauses, this set would include the corresponding clause \eqref{eq:undefined}. However, $\varphi(x_{i,t})=1$ by the construction of $w$ and $\varphi(y_{j,t-1})=1$ by Claim~2. Hence $\varphi$ does not satisfy this clause while the clauses from $T$ are satisfied by $\varphi$, a contradiction.
 	
 	Finally, consider $(C3)$. By Claim~1, some state is active after $\ell$ steps. On the other hand, the assignment $\varphi$ satisfies the clauses in $S$, which means that $\varphi(y_{j,\ell})=1$ for at most one index $j\in\{1,\dots,n\}$. By Claim~2 this implies that at most one state is active after $\ell$ steps. Thus, exactly one state is active after $\ell$ steps, that is, $|Q\dt w|=1$.
 \end{proof}

\textbf{Remark 2.}  In~\cite{Shiop}, the first author has constructed a SAT-encoding for the problem ESW (the existence of a \esw\ of a given length), which is defined analogously to CSW. We will not reproduce this encoding here; it uses the same set of variables as the above encoding for CSW but the set of clauses is essentially different. One may think that since the definition of an \esw\ differs from the definition of a \csw\ by the absence of the conditions $(C1)$ and $(C2)$, one could get an encoding for ESW by just omitting the clauses~\eqref{eq:undefined} that control these conditions in the encoding for CSW, and vice versa, one could encode CSW by appending the clauses~\eqref{eq:undefined} to the encoding for ESW. However, it is easy to exhibit counterexamples to show that such a naive transformation of our CSW encoding into an encoding for ESW fails. In the converse direction, the transformation produces a valid encoding for CSW but this encoding has many more clauses than the CSW encoding suggested here.

\smallskip

\textbf{Remark 3.} Observe that all clauses in the group $S$ of synchronization clauses and a majority of clauses in the group $L$ of letter clauses are typical `at-most-one' constraints. There are various way to express such constraints by fewer clauses. In our implementation, we have used the so-called ladder encoding suggested in \cite{GN04}, see also~\cite[Chapter~2]{BHvMW}. We demonstrate how the ladder encoding works on the set $S$.
We introduce $n-1$ additional variables  $f_{1},f_{2},\dots,f_{n-1}$ and substitute the clauses~\eqref{eq:only one state} by two new groups of clauses: the ladder validity clauses
\begin{equation}
\label{eq:ladder}
\neg f_{j+1}\vee f_{j}
\end{equation}
for $j=1,2,\dots,n-2$,
and the channelling clauses that correspond to the equivalence $y_{j,\ell}\longleftrightarrow f_{j-1}\mathop{\&}\neg f_{j}$:
\begin{equation}
\label{eq:channel}
\neg f_{j-1}\vee f_{j}\vee y_{j,\ell},\quad \neg y_{j,\ell}\vee f_{j-1},\quad \neg y_{j,\ell}\vee \neg f_{j}
\end{equation}
for $j=1,2,\dots,n$, where the clauses containing $f_0$ or $f_n$ are simplified as if $f_0=1$ or $f_n=0$. Altogether, we get $4n-4$ clauses in \eqref{eq:ladder} and \eqref{eq:channel} instead of $\frac{n(n-1)}2$ clauses in \eqref{eq:only one state} on the price of adding $n-1$ extra variables. The same trick allows us to decrease the number of clauses in the set $L$, but this is less important because the parameter $m$ (the size of the input alphabet) is usually small. Our experiments have shown that using ladder encoding significantly reduces time needed to solve CSW instances, especially for automata with large number of states.

\smallskip

\textbf{Remark 4.} In a majority of our experiments, we deal with PFAs that have only two input letters. Let us call such PFAs \emph{binary}. To encode the CSW/ESW instance $(\mathrsfs{A},\ell)$, where $\mA$ is a binary PFA, we can use only $\ell$ letter variables $x_1,\dots,x_\ell$ to encode the letters of a hypothetical carefully/exactly \sw\ $w$ of length $\ell$ since there is an obvious 1-1 correspondence between the truth assignments on the set $\{x_1,\dots,x_\ell\}$ and the words of length $\ell$ over any fixed 2-letter alphabet. In more formal terms, we can modify the above encoding of CSW and the corresponding encoding of ESW from \cite{Shiop}, substituting $x_t$ for $x_{1,t}$ and $\neg x_t$ for $x_{2,t}$ for all $t=1,2,\dots,\ell$ in all clauses in which $x_{1,t}$ or $x_{2,t}$ occur. Observe that the letter clauses \eqref{eq:only one letter} become tautologies after this substitution, and hence, they can be safely omitted. Thus, for a binary PFA $\mA$ with $n$ states, we may encode the CSW instance $(\mathrsfs{A},\ell)$ into a SAT instance with $\ell(n+2)+n-2$ variables and only $2\ell n+5n-4$ clauses if we use both the modification just described and the ladder encoding from Remark~3.

\section{Design of experiments}
\label{exp}

\subsection{General framework}
\label{general-exp-pfa}
The general framework of our experiments with random automata consists of the following basic steps.
\begin{enumerate}
\itemsep=-4pt
	\item[1.] A positive integer $n$ (the number of states) is fixed.
	\item[2.] A random PFA $\mathrsfs{A}$ with $n$ states is generated. 	
	\item[3.] The pair $(\mathrsfs{A},1)$ is encoded into a SAT instance $(V',C')$ as described in Sect.~\ref{sec:SAT} (if we study careful synchronization) or in \cite{Shiop} (if we study exact synchronization).
	\item[4.] The instance $(V',C')$ is scaled to the instance $(V,C)$ that encodes the pair $(\mathrsfs{A},\ell)$, see Remark~5 below.
	\item[5.] The SAT solver MiniSat 2.2.0 is invoked to solve the SAT instance $(V,C)$.
\end{enumerate}
We refer to~\cite{Minisat} for a description of the underlying ideas of MiniSat and to~\cite{Minisat_page} for a discussion and the source code of the solver.

\smallskip

\textbf{Remark 5.} An important feature of our encodings is that as soon as we have constructed the `primary' SAT instance $(V',C')$ that encodes the CSW/ESW instance $(\mathrsfs{A},1)$, we are in a position to scale $(V',C')$ to the SAT instance encoding the CSW/ESW instance $(\mathrsfs{A},\ell)$ for any value of $\ell$. In order to explain this feature, recall that MiniSAT accepts its input in the following text format (so-called simplified DIMACS CNF format). Every line beginning with \texttt{c} is a comment. The first non-comment line is of the form:

\texttt{p cnf NUMBER\_OF\_VARIABLES NUMBER\_OF\_CLAUSES}

\noindent Variables are represented by integers from 1 to \texttt{NUMBER\_OF\_VARIABLES}. The first non-comment line is followed by \texttt{NUMBER\_OF\_CLAUSES} non-comment lines each of which defines a clause. Every such line starts with a space-separated list of different non-zero integers corresponding to the literals of the clause: a positive integer corresponds to a literal which is a variable, and a negative integer corresponds to a literal which is the negation of a variable; the line ends in a space and the number~0.

For simplicity, we describe the scaling procedure for binary PFAs only and we assume that the ladder encoding from Remark~4 has not been used. (Both the generalization to PFAs over larger alphabets and the modification needed to accommodate additional variables involved in the ladder encoding are fairly straightforward.) Given a binary PFA $\mathrsfs{A}$ with $n$ states, we write the SAT instance $(V',C')$, which corresponds to $(\mathrsfs{A},1)$, in DIMACS CNF format, representing the variables $x_1$, $y_{j,0}$, $y_{j,1}$, $j=1,\dots,n$, by the numbers, respectively, $n+1$, $j$, $j+n+1$. For an illustration, see Table~\ref{tb:encoding} that shows the SAT encoding for the PFA $\mathrsfs{P}_4$ from Fig.~\ref{fig:example}.

\begin{table}[ht]
	\caption{The SAT encoding of the CSW instance $(\mathrsfs{P}_4,1)$}\label{tb:encoding}
	\begin{center}
		\begin{tabular}{|rp{5cm}|p{5cm}|}
			\hline
			& \centerline{\raisebox{-8pt}{\textbf{Clauses}}\rule{30pt}{0pt}}  &  \centerline{\raisebox{-8pt}{\textbf{DIMACS CNF lines}}} \\
			\hline
			&& \begin{tabular}{l}\texttt{p cnf 9 18}\end{tabular}\\[4pt]
			$I'$ &$\left\{\begin{tabular}{l}
			$y_{1,0}$\\
			$y_{2,0}$\\
			$y_{3,0}$\\ 	
			$y_{4,0}$				
			\end{tabular}\right.$
			&\begin{tabular}{l}
				\texttt{1 0}\\
				\texttt{2 0}\\
				\texttt{3 0} \\
				\texttt{4 0} 			
			\end{tabular}\\
			&&\\[-5pt]
			$T'$ &$\left\{\begin{tabular}{l}
			$\neg y_{1,0}\vee \neg x_1 \vee y_{2,1}$\\
			$\neg y_{2,0}\vee \neg x_1 \vee y_{2,1}$\\
			$\neg y_{3,0}\vee \neg x_1 \vee y_{3,1}$\\
			$\neg y_{4,0}\vee \neg x_1 \vee y_{1,1}$\\
			$\neg y_{1,0}\vee x_1 \vee y_{2,1}$\\
			$\neg y_{2,0}\vee x_1\vee y_{3,1}$\\	
			$\neg y_{3,0}\vee x_1\vee y_{4,1}$\\
			$\neg y_{4,0}\vee x_1$			
			\end{tabular}\right.$
			&\begin{tabular}{l}
				\texttt{-1 -5 7 0}\\
				\texttt{-2 -5 7 0}\\
				\texttt{-3 -5 8 0}\\
				\texttt{-4 -5 6 0}\\
				\texttt{-1  5 7 0}\\
				\texttt{-2  5 8 0}\\
				\texttt{-3  5 9 0}\\
				\texttt{-4  5 0}	
			\end{tabular}\\
			&&\\[-5pt]
			$S'$ &$\left\{\begin{tabular}{l}
			$\neg y_{1,1}\vee \neg y_{2,1}$\\
			$\neg y_{1,1}\vee \neg y_{3,1}$\\	
			$\neg y_{1,1}\vee \neg y_{4,1}$\\
			$\neg y_{2,1}\vee \neg y_{3,1}$\\
			$\neg y_{2,1}\vee \neg y_{4,1}$\\
			$\neg y_{3,1}\vee \neg y_{4,1}$\\
			\end{tabular}\right.$
			&\begin{tabular}{l}
				\texttt{-6 -7   0}\\
				\texttt{-6 -8   0}\\
				\texttt{-6 -9   0}\\
				\texttt{-7 -8   0}\\
				\texttt{-7 -9   0}\\
				\texttt{-8 -9   0}			
			\end{tabular}\\[-5pt]
             &&\\
			\hline
		\end{tabular}
	\end{center}
\end{table}

Now, to get the SAT instance $(V,C)$ that encodes the pair $(\mathrsfs{A},\ell)$ for some $\ell>1$, one transforms the DIMACS CNF representation of $C'=I'\cup T'\cup S'$ as follows:

\begin{enumerate}
\itemsep=-4pt
	\item[1.] In the first non-comment line, replace the numbers $2n+1$ and $2n+\tfrac{n(n+1)}2$ by respectively $(\ell+1)n+\ell$ and $2\ell n+\tfrac{n(n+1)}2$.
	\item[2.] Keep the lines corresponding to the clauses in $I'$ and $T'$.
	\item[3.] For each $t=2,\dots,\ell$, add all the lines obtained from the lines that correspond to the clauses in $T'$ by keeping the sign of every non-zero integer and adding $(t-1)(n+1)$ to its absolute value.
	\item[4.] In each line corresponding to a clause in $S'$, substitute every nonzero integer $\pm k$ by the integer $\pm(k+(\ell-1)(n+1))$.
\end{enumerate}

\subsection{Types of experiments and implementation details}
\label{design-exp-pfa}

We have performed four series of experiments with random PFAs.
\begin{enumerate}
\itemindent=26pt
	\item[Series 1:] studying the probability of being synchronizing  for each version of synchronization for randomly generated binary PFAs with one undefined transition.
	\item[Series 2:] finding an approximation for the average length of shortest \csws\  and \esws\ for randomly generated binary PFAs with one undefined transition.
	\item[Series 3:] studying the influence of the input alphabet size on the length of the shortest synchronizing word. 		
	\item[Series 4:] studying the influence of the \emph{density} (the number of defined transitions) on the length of the shortest synchronizing word.
\end{enumerate}

All our algorithms were implemented in C++ and compiled with GCC 4.9.2. In our experiments we used a laptop with an Intel(R) Core(TM) i5-2520M processor with 2.5 GHz CPU and 4GB of RAM. Our code and datasets are available under \url{https://github.com/hananshabana/SynchronizationChecker}.

\subsection{Generating random PFAs for experiments in Series 1--4}
\label{subsect:generating}
In experiments from Series 1 and 2 we worked with binary PFAs $\mathrsfs{A}=\langle Q,\Sigma\rangle$ with $n\le 100$ states and only one undefined transition. Then one letter must be everywhere defined; we denoted it by $a$ and selected the action of $a$ uniformly at random from all $n^{n}$ maps $Q\to Q$. To ensure that there is a unique undefined transition with $b$, we chose uniformly at random a state $q_b\in Q$ and then selected the action of $b$ uniformly at random from all $n^{n-1}$ maps $Q\setminus\{q_b\}\to Q$.

In experiments from Series 3, we again considered PFAs with an everywhere defined letter and defined its action as above. Then, for each of the remaining letters, we first chose the number $k$ of states at which the letter should be defined; $k$ was chosen uniformly at random from the set $\{1,2,\dots,n-1\}$. Then we selected uniformly at random $k$ different states from $Q$, and for each of these states we chose uniformly at random a state in $Q$ as its image under the action of the letter.

In experiments from Series 4, we considered binary PFAs $\mathrsfs{A}$ with $n\le 100$ states. Let $\rho_{\mathrsfs{A}}$ stand for the density of $\mathrsfs{A}$. In experiments with careful synchronization, possible values of $\rho_{\mathrsfs{A}}$ were chosen between $n+1$ and $2n-1$ since one of the letters must be everywhere defined. For the other letter, we set $k:=\rho_{\mathrsfs{A}}-n$ and then proceeded as in the preceding paragraph. In experiments with exact synchronization, the value of $\rho_{\mathrsfs{A}}$ can be any number between $1$ and $2n-1$. However, it is easy to realize that a PFA with density 1, that is, a PFA with a unique defined transition, is always exactly synchronizing by a word of length 1 (namely, by the letter which action at some state is defined). Thus, the case of density 1 is not interesting at all, and we chose possible values of $\rho_{\mathrsfs{A}}$ from numbers between $2$ and $2n-1$. Then we chose uniformly at random a non-negative number $k\le\rho_{\mathrsfs{A}}$ and applied the procedure described in the preceding paragraph first to $k$ and then to $\rho_{\mathrsfs{A}}-k$.

\section{Experimental results for randomly generated automata and their analysis}
\label{samples}
\subsection{Series 1: Probability of synchronization}
\label{subsect:probability}
This series of our experiments aims to compare the probability of being exactly or carefully synchronizing for the same sample of random automata. Figure~\ref{prop-carf-exact} shows the probability of being synchronizing in each of these two versions of synchronization for the class of binary PFAs $\langle Q,\{a,b\}\rangle$ such that the letter $a$ is everywhere defined and the letter $b$ is undefined at exactly one state in $Q$. For brevity, we refer to automata from this class as \emph{almost complete} PFAs.  Observe that the problem of deciding whether or not a given PFA is carefully synchronizing remains PSPACE-complete even if restricted to this rather special case~\cite[Theorem~3]{Martyugin12}. For each fixed~$n$, we generated up to 1000 random almost complete PFAs.

\begin{figure}[htb]
	\centering
	\begin{tikzpicture}[scale=0.9]
	\begin{axis}[
	xlabel= {Number of states $n$},
	ylabel={Probability},
	legend pos=south east
	]
	\addplot+[mark=*  ]
	plot coordinates {
		
		(5,0.416)
		(10,0.61)
		(20,0.706)
		(30,0.776)
		(40,0.816)
		(50,0.82)
		(60,0.854)
		(70,0.856)
		(80,0.862)
		(100,0.871)
		
	};
	\addlegendentry{Careful synchronization}
	
	\addplot+[mark=*]
	plot coordinates {
		(5,0.936)
		(10,0.964)
		(20,0.996)
		(30,0.998)
		(40,1)
		(50,0.998)
		(60,1)
		(70,1)
		(80,0.998)
		(100,1)
	};
	\addlegendentry{Exact synchronization}
	\end{axis}
	\end{tikzpicture}
	\caption{Probability of being synchronizing for two versions of  synchronization}\label{prop-carf-exact}
\end{figure}
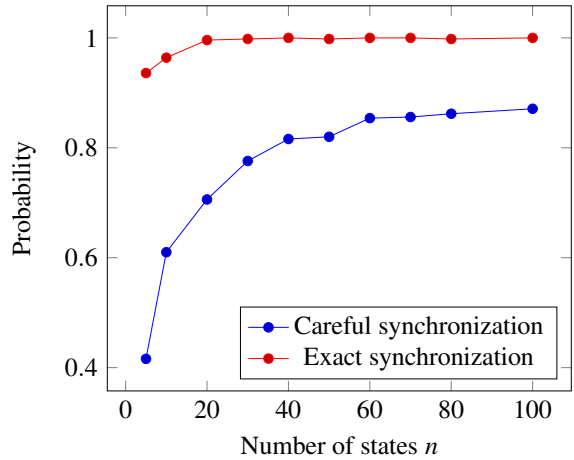

Let $P_{E}(n)$ stand for the probability of a random almost complete PFA with $n$ states to be exactly synchronizing and let $P_{C}(n)$ be the probability that the same random PFA is carefully synchronizing. We have $P_{E}(n) > P_{C}(n)$ since, as mentioned, every carefully synchronizing PFA is exactly synchronizing. The data in Fig.~\ref{prop-carf-exact} show that the gap between $P_{E}(n)$ and $P_{C}(n)$ decreases as $n$ grows but remains non-negligible ever for $n$ close to 100. We also see that $P_{E}(n)$ quickly tends to 1 as the state number grows. Recall the same effect was experimentally observed for CFAs and was then theoretically justified by Berlinkov~\cite{Berlinkov:2016} and Nicaud~\cite{Nicaud:2016,Nicaud:2019}: the probability $P_S(n)$ that a random binary CFA with $n$ states is synchronizing tends to 1 as $n$ tends to infinity. Moreover, it is shown in~\cite{Berlinkov:2016} that $1-P_S(n)=\varTheta(\frac1n)$. It is not difficult to extend the latter result to random almost complete PFAs. Since we have not found such an extension in the literature, we have included it here.

The extension is based on the following easy observation.
\begin{lem}
\label{lem:reduction}
Let $\mA$ be a synchronizing CFA. Then every PFA obtained from $\mA$ by removing a single transition is exactly synchronizing.
\end{lem}

\begin{proof}
Let $\mA=\langle Q,\Sigma,\delta\rangle$. Fix an arbitrary pair $(p,b)\in Q\times\Sigma$ and consider the PFA $\mB=\langle Q,\Sigma,\zeta\rangle$, where $\zeta$ coincides with $\delta$ on the set $Q\times\Sigma\setminus\{(p,b)\}$ and $\zeta(p,b)$ is undefined. Let $w\in\Sigma^*$ be such that $|\delta(Q,w)|=1$. Clearly, if $w$ is defined in $\mB$ at some $q\in Q$, then $\zeta(q,w)=\delta(q,w)$ whence $|\zeta(Q,w)|=1$ and $w$ is an \esw\ for $\mB$. Thus, assume that $w$ is nowhere defined in $\mB$. Let $v$ be the longest prefix of $w$ which is defined in $\mB$ at some state and let $x$ be the letter that follows $v$ in $w$. By the choice of $v$, the set $P:=\zeta(Q,v)$ is not empty but $\zeta(P,x)$ is empty. Thus, all transitions of the form $\zeta(q,x)$ with $q\in P$ must be undefined. However, by the definition of $\zeta$, the only undefined transition in $\mB$ is $\zeta(p,b)$ whence $x=b$ and $P=\{p\}$. In particular, $|\zeta(Q,v)|=1$ and $v$ is an \esw\ for $\mB$.
\end{proof}

\begin{pro}
\label{prop:exact-syn-bound}
$1-P_{E}(n)= \varTheta(\frac1{n})$.
\end{pro}

\begin{proof}
Given an almost complete PFA $\mB=\langle Q,\{a,b\},\zeta\rangle$, its \emph{completion} is any CFA obtained by defining the undefined transition of $\mA$. Let $n:=|Q|$. If the letter $b$ is undefined at a certain state $p\in Q$, we can choose any state in $Q$ as the image of $p$ under $b$ in the completion, whence $\mB$ has $n$ different completions. Conversely, any CFA $\mA=\langle Q,\{a,b\},\delta\rangle$ serves as a completion for $n$ different almost complete PFAs obtained from $\mA$ by removing the transition $\delta(p,b)$, where $p$ runs over $Q$.

Now consider the set $\mathbf{P}$ of all pairs $(\mB,\mA)$ such that $\mA$ is a completion of $\mB$ and $\mB$ is not exactly synchronizing. Denoting by $N$ is the number of almost complete PFAs with the state set $Q$ that are not exactly synchronizing, we have $|\mathbf{P}|=Nn$. Lemma~\ref{lem:reduction} implies that no CFA $\mA$ such that there is $\mB$ with $(\mB,\mA)\in\mathbf{P}$ can be synchronizing. Any non-synchronizing CFA may occur in at most $n$ pairs from $\mathbf{P}$ whence the number $M$ of CFAs with the state set $Q$ that are not synchronizing satisfies $M\ge\dfrac{|\mathbf{P}|}n=\dfrac{Nn}n=N$. Observe that the total number $n^{2n}$ of binary CFAs is the same as the total number of almost complete PFAs: to construct an an almost complete PFA with $n$ states, we have $n^n$ choices for the action of the everywhere defined letter, $n$ choices for a state at which the other letter is undefined, and $n^{n-1}$ choices for the action of the latter letter at the remaining $n-1$ states. Therefore, we conclude that
\[
1-P_{E}(n)=\frac{N}{n^{2n}}\le\frac{M}{n^{2n}}=1-P_S(n)=O(\frac1n).
\]
Hence, $1-P_E(n)=O(\frac1n)$.

In order to get a matching lower bound for $1-P_E(n)$, we describe a construction that yields `sufficiently many' almost complete PFAs with $n$ states and 2 letters $a$ and $b$ that are not exactly synchronizing. The construction is as follows. First we choose a state $q_0$ at which $b$ is undefined. There are $n$ choices for $q_0$. Then we define the action of $a$ at $q_0$ in an arbitrary way. This gives $n$ choices. After that, there are $n-1$ choices for the state $q_1$ which is fixed by both $a$ and $b$. Finally, there are $(n-2)^{2(n-2)}$ choices for the actions of $a$ and $b$ at the remaining $n-2$ states.  Altogether, the construction gives  $n^{2}(n-1)(n-2)^{2(n-2)}$  almost complete automata, and it is clear that none of PFAs constructed this way are exactly synchronizing. Now when we calculate the fraction $n^{2}(n-1)(n-2)^{2(n-2)}/n^{2n}$, we get
	\[
	\frac{n^{2}(n-1)(n-2)^{2(n-2)}}{n^{2n}}=\left(1-\frac1n\right)\left(1-\frac2n\right)^{2n}\left(1-\frac2n\right)^{-4}\frac1{n}.
	\]
As $n$ tends to the infinity, the first and the third factors tend to 1, and the second factor tends to $e^{-4}$. Thus, the fraction is asymptotically equivalent to $\frac{e^{-4}}{n}$. Hence $1-P_{E}(n)= \varOmega(\frac1{n})$.
\end{proof}

Back to Fig.~\ref{prop-carf-exact}, we see that the probability $P_{C}(n)$ also grows with $n$ but it not clear if it tends to 1 as $n$ tends to infinity. To the best of our knowledge, no theoretical results published so far predict the asymptotic behavior of the function $P_{C}(n)$ nor, more generally, the asymptotic behavior of the probability of being carefully synchronizing for any class of random PFAs. Here, as a result of analysis of the outcome of our experiments, we are able to show that even if $P_{C}(n)$ approaches 1 as $n\to\infty$, it does it at much slower rate than $P_{E}(n)$; see the discussion at the end of the subsection.

First, let us discuss how we proceeded to determine if a PFA $\mA$ from our sample was carefully/exactly synchronizing. According to the general scheme described in Subsect.~\ref{general-exp-pfa}, we encoded $(\mA,1)$ as a SAT instance, wrote the instance in DIMACS CNF format, and scaled it to the instances encoding $(\mA,\ell)$ with $\ell=2,4,8,\dots$ until we reached an instance on which the SAT solver returned YES. Of course, sometimes it  happened that we did not reach such an instance which indicated that either $\mA$ was not carefully/exactly synchronizing or the minimum length of carefully/exactly \sws\ for $\mA$ was too big so that MiniSat 2.2.0 could not handle the resulting SAT instance. In such cases, we had to use some additional ideas to distinguish between non-synchronizing and `too slowly' synchronizing automata.

For exact synchronization, an additional analysis was needed only for small values of $n$ ($n\le 20$) and for a few exceptional PFAs with $n>30$. We analyzed these cases using a brute force algorithm known as the successor tree method. See the recent paper by T\"urker~\cite{Turkish19} for a description of the method and its modern implementation\footnote{T\"urker~\cite{Turkish19} uses the term `reset sequence' for what we call `exactly synchronizing word'.}.

The situation for careful synchronization was more involved. The only known brute force algorithm for careful synchronization is the partial power automaton method, which we will discuss (and compare with our approach) in Sect.~\ref{sec:power}. It turned out that this method could hardly handle PFAs with more than 20 states. Therefore, we devised a simple theoretical condition under which a binary PFA is not carefully synchronizing and checked PFAs against this condition, prior to having started the procedures from Subsect.~\ref{general-exp-pfa}.

Let $q$ be a state and $a$ letter of a PFA. We say that $q$ is $a$-\emph{cyclic} if $q=q.a^k$ for some positive integer $k$.
\begin{lem}
\label{lem:cyclic}
Let a PFA $\mA=\langle Q,\{a,b\}\rangle$ be such that the letter $a$ is everywhere defined and has at least two $a$-cyclic states. If the letter $b$ is undefined at some $a$-cyclic state, the PFA $\mA$ is not carefully synchronizing.
\end{lem}

\begin{proof}
Arguing by contradiction, suppose that $w\in\{a,b\}^*$ is a \csw\ for $\mA$. Then $w$ starts with $a$ because of the condition $(C1)$. Further, $w$ cannot be a power of $a$ because of the condition $(C3)$ as $a$ has at least two $a$-cyclic states and each $a$-cyclic state belongs to the image of an arbitrary power of $a$. Thus, the letter $b$ occurs in $w$ whence $w$ has a prefix of the form $a^sb$ for some positive integer $s$. As mentioned, each $a$-cyclic state belongs to $Q.a^s$, and we get a contradiction with the condition $(C2)$ as $b$ is undefined at some $a$-cyclic state.
\end{proof}

Clearly, given a binary PFA, it is easy to verify if the PFA satisfies the premises of Lemma~\ref{lem:cyclic}. It is Lemma~\ref{lem:cyclic} that we used to filter out almost complete PFAs that were not carefully synchronizing before having run the SAT-solver method. We stress that Lemma~\ref{lem:cyclic} is only a sufficient condition for an almost complete PFA to be not carefully synchronizing. However, it was well suited for our purposes because it turned out to be applicable frequently enough. Indeed, the statistical properties of random maps are well studied; in particular, if the random variable $\xi$ represents the number of cyclic points of a map chosen uniformly at random from all $n^n$ maps on an $n$-element set, the following expression for the probability of the event $\xi=j$, where $j\in\{1,2,\dots,n\}$, is known (see~\cite{Harris}):
\begin{equation}
\label{eq:harris}
P(\xi=j)=\frac{(n-1)!j}{(n-j)!n^j}.
\end{equation}
For the premises of Lemma~\ref{lem:cyclic} to hold for an almost complete PFA $\mA=\langle Q,\{a,b\}\rangle$, the map $Q\to Q$ induced by the letter $a$ must have at least two cyclic points (= $a$-cyclic states), and the only state at which the letter $b$ is undefined must be $a$-cyclic. Denoting $|Q|$ by $n$, we derive from~\eqref{eq:harris} the following expression for the probability that Lemma~\ref{lem:cyclic} applies to $\mA$:
\begin{equation}
\label{eq:expectation}
\sum_{j=2}^{n-1}\frac{j}{n}P(\xi=j)=\sum_{j=2}^{n-1}\frac{(n-1)!j^2}{(n-j)!n^{j+1}}.
\end{equation}
Observe that the expression~\eqref{eq:expectation} differs in just one summand $\frac{1}{n}P(\xi=1)=\dfrac1{n^2}$ from
\[
n^{-1}E[\xi]=\sum_{j=1}^{n}\frac{j}{n}P(\xi=j).
\]

Evaluating the expression~\eqref{eq:expectation} at $n=100$, say, one gets 0.121989414. (For the numerical computations, we used an elegant method suggested by Zubkov~\cite{Zubkov}.) Thus, more than 12\% of randomly chosen almost complete PFAs with 100 states satisfy the premises of Lemma~\ref{lem:cyclic}. On the other hand, the SAT-solver approach in our experiments succeeded for more than 87\% of almost complete PFAs with 100 states. It is what we meant above when having said that Lemma~\ref{lem:cyclic} was well sufficient to confirm the absence of careful synchronization for an overwhelming majority of almost complete PFAs which are not carefully synchronizing, and thus, to avoid the SAT-solver having to work in vain.

Back to the aforementioned question of the asymptotic behavior of the function $P_{C}(n)$, we notice that even though Lemma~\ref{lem:cyclic} does not exclude $P_{C}(n)$ tending to 1, it allows us to show that even if $P_{C}(n)$ tends to 1 as $n\to\infty$, the convergence rate should be relatively slow. Indeed, it is known (see~\cite{Harris}) that the expectation $E[\xi]$ is asymptotically equivalent to $\sqrt{\dfrac{\pi n}2}$. As observed, the probability \eqref{eq:expectation} that Lemma~\ref{lem:cyclic} applies to a random almost complete PFAs with $n$ states differs from $n^{-1}E[\xi]\sim\sqrt{\dfrac{\pi}{2n}}$ by $\dfrac1{n^2}$, which is asymptotically negligible in comparison with $\sqrt{\dfrac{\pi}{2n}}$. By Lemma~\ref{lem:cyclic}, we have that the difference $1-P_{C}(n)$, that is, the probability that an almost complete PFAs with $n$ states is not carefully synchronizing is asymptotically greater than or equivalent to $\sqrt{\dfrac{\pi}{2n}}$. Thus, $1-P_{C}(n)=\Omega(\frac1{\sqrt{n}})$, while we have demonstrated above that $1-P_E(n)=\varTheta(\frac1{n})$.

\subsection{Series 2: Average length of shortest \csws}
\label{subsect:length}

Here we present only results obtained in the case of careful synchronization since our parallel results for exact synchronization have already been reported in~\cite{Shiop}.

\begin{figure}[H]
	\centering
	\begin{tikzpicture}
	\begin{axis}[
	xlabel= {Number of states $n$},
	ylabel={$\ell_C(n)$},
	domain=10:100,
	legend pos=north west
	]
	\addplot+[mark=*]
	plot coordinates {
		(10,7.480)
		(15,9.790)
		(17,10.680)
		(20,11.610)
		(23,12.580)
		(25,13.150)
		(28,14.010)
		(30,14.410)
		(35,15.660)
		(40,16.770)
		(45,17.790)
		(50,18.870)
		(60,20.600)
		(70,22.220)
		(80,23.820)
		(90,25.040)
		(100,26.550)
		(55,19.750)
		(65,21.280)
		
	};
	\addlegendentry{Observed}
	
	\addplot +[mark=.][domain=10:100] {3.92+(0.49*x)+(-0.005*x^2)+(0.000024*x^3)};
	
	\addlegendentry{Our estimation}
	\end{axis}
	\end{tikzpicture}
	\caption{Approximation of the average length of shortest \csws\ for carefully synchronizing almost complete PFAs with $n$ states}\label{approx-carf}
\end{figure}
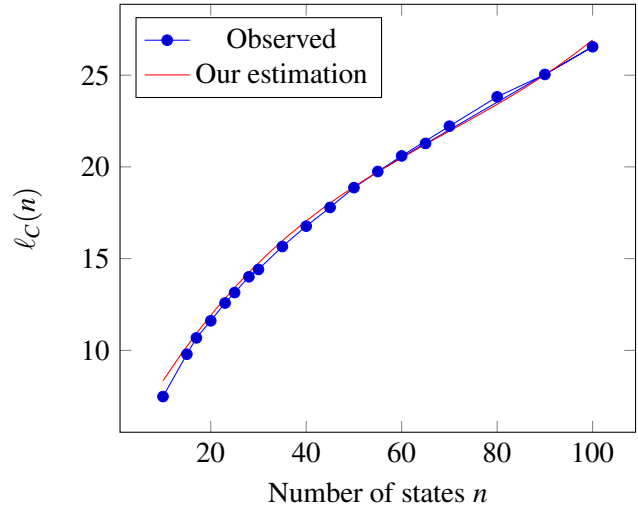

We worked with almost complete PFAs that were found to be carefully synchronizing in the course of the experiment detailed in Subsect.~\ref{subsect:probability}. For such a PFA $\mA$, we were left at the end of the experiment with a number $\ell$, the least power of 2 for which MiniSat returns YES  on the SAT instance that encodes the CSW instance $(\mA,\ell)$.  In order to find a \csw\ of minimum length for $\mA$, we performed standard binary search, having started with $\ell_{\max}:=\ell$ and $\ell_{\min}:=\dfrac\ell2$. That is, we
\begin{enumerate}
\item[1)] let $\ell:=\dfrac{\ell_{\min}+\ell_{\max}}2$;
\item[2)] run MiniSat on the SAT instance that encodes the CSW instance $(\mA,\ell)$;
\item[3)] let $\ell_{\max}:=\ell$ if the answer returned by Minisat was YES, and let $\ell_{\min}:=\ell$ if the answer was NO;
\item[4)] check if $\ell_{\max}-\ell_{\min}=1$: YES means that $\ell_{\max}$ is the minimum length of \csws\ for $\mathrsfs{A}$; NO means that we have to return to Step 1).
\end{enumerate}

Using experimental data found this way, we calculated the average length $\ell_C(n)$ for shortest \csws\ of carefully synchronizing almost complete PFAs with $n$ states. Then we used the least squares method to find a function that best reflects how $\ell_C(n)$ depends on $n$. It turned out that our results are reasonably well approximated by the following expression:
\begin{equation}
\label{eq-approxcarf}
\ell_C(n)\approx 3.92+0.49n-0.005n^2+0.000024n^3.
\end{equation}
We mention that the results for exact synchronization in~\cite{Shiop} look quite similar.

\begin{figure}[H]
	\centering
	\begin{tikzpicture}
	\begin{axis}[
	xlabel= {Number of states },
	ylabel={Relative standard deviation},
	legend pos=outer north east
	]
	\addplot+[mark=*]
	plot coordinates {
		(10,0.38159225)
		(15,0.318618803)
		(17,0.304359789)
		(20,0.272500964)
		(23,0.264849967)
		(25,0.255621131)
		(28,0.236635543)
		(30,0.235014587)
		(35,0.220527863)
		(40,0.208323052)
		(45,0.200001709)
		(50,0.192410953)
		(55,0.187375028)
		(60,0.185061504)
		(65,0.167199559)
		(70,0.171558502)
		(80,0.167102896)
		(90,0.161681628)
		(100,0.152451951)
	};
	\end{axis}
	\end{tikzpicture}
	\caption{Relative standard deviation of datasets}\label{rsdcarful}
\end{figure}
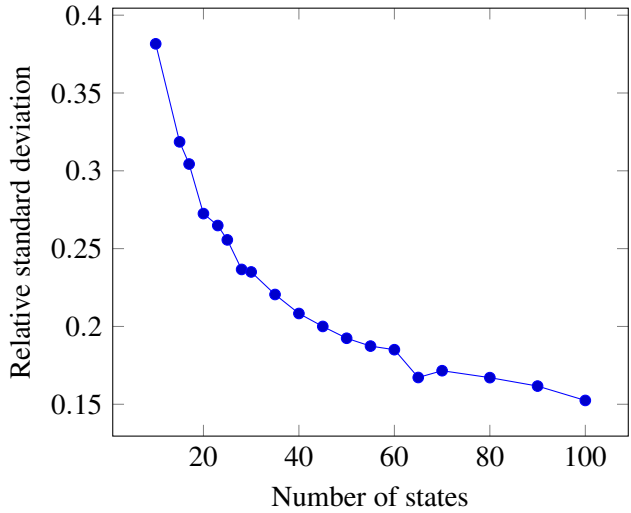

The relation between the approximation~\eqref{eq-approxcarf} and our experimental data is shown in Fig.~\ref{approx-carf}, while Fig.~\ref{rsdcarful} shows the relation between the relative standard deviation of our datasets and the number of states. We see that the relative standard deviation gradually decreases as the number of states grows.

\subsection{Series 3: Influence of the input alphabet size}

Here again, we report only results obtained in the case of careful synchronization. This series of experiments aimed to see how the length of the shortest carefully synchronizing word is affected by the number of input letters.
We experimented with samples of carefully synchronizing PFAs with varying state and input alphabet sizes but approximately the same \emph{relative density}, that is, the same ratio between the density and the number of states. We generated random PFAs as described in Subsect.~\ref{subsect:generating} and applied Lemma~\ref{lem:cyclic} for filtering out PFAs that were not carefully synchronizing. Then we used binary search as in Subsect.~\ref{subsect:length} to determine the minimum length of \csws.

\begin{figure}[H]
	\centering
	\begin{tikzpicture}
	\begin{axis}[
	xlabel= {Number of states $n$ },
	ylabel={Mean length of shortest \csws},
	legend pos=  north west
	]
	\addplot+[mark=., smooth]
	plot coordinates {
		(30,10.5)
		(50,13.7)
		(60,15.05)
		(70,16.3)
		(80,17.8)
		
	};
	\addlegendentry{$|\Sigma|= 3,\ \rho=2n-1$ }
	
	\addplot+[mark=.]
	plot coordinates {
		(30,14.410)  	
		(50,18.870)
		(60,20.600)
		(70,22.220)
		(80,23.820)
			
	};
	\addlegendentry{$|\Sigma|= 2,\ \rho=2n-1 $}
	\end{axis}
	\end{tikzpicture}
	\caption{The cardinality of the input alphabet versus the length of the shortest synchronizing word}
	\label{Cardinality-length}
\end{figure}
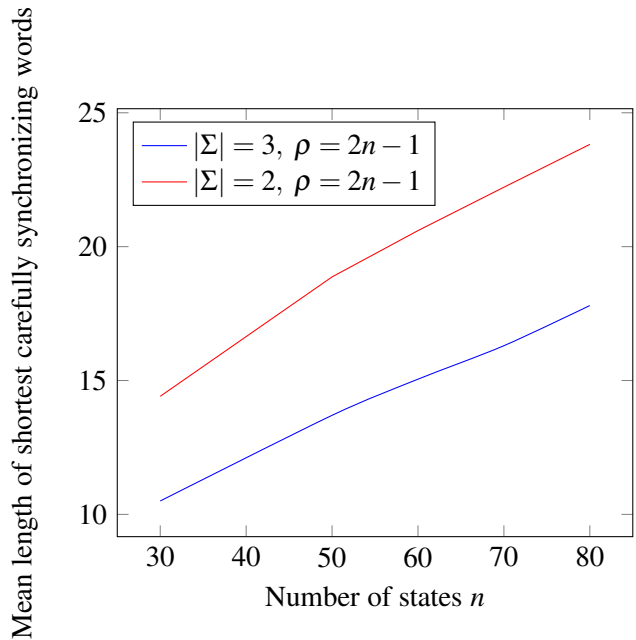

Figure~\ref{Cardinality-length} may serve as an illustration for typical results found in this series of experiments. It shows the average lengths of shortest \csws\ for carefully synchronizing PFAs with 2 or 3 input letters and relative density close to 2. More precisely, we considered PFAs with $n$ states and the density $\rho=2n-1$. (Thus, in the case of 2 input letters, we dealt with almost complete PFAs so that we were in a position to partly re-use the data computed in experiments in Subsect.~\ref{subsect:length}.) We see that the corresponding graphs have similar regular shape and that PFAs with a larger input alphabet synchronize faster. These conclusions held also when other values of relative density were fixed. The observed phenomena are intuitively plausible as having more letters gives more degrees of freedom for careful synchronization and it is to expect that  \csws\ become shorter. However, we have got no rigorous theoretical explanations for these phenomena so far.

\begin{figure}[p]
	\centering
	\begin{tikzpicture}
	
	\begin{axis}[
	xlabel= {$\rho$},
	ylabel={Mean length of shortest \esws},
	legend pos=north east
	]
	\addplot+[mark=.]
	plot coordinates {
		(3,1.085)
		(5,1.26)
		(11,2.03)
		(17,2.26)
		(25,2.9)
		(30,3.26)
		(35,3.86)
		(40,4.26)
		(50,6.23)
		(55,8.83)
		(59, 12.85)
	};

	\end{axis}
	\end{tikzpicture}
	\caption{Exact synchronization versus density for $30$ states}\label{exact-density}

\bigskip

	\centering
	\begin{tikzpicture} 	
	\begin{axis}[
	xlabel= { $\rho$},
	ylabel={Mean length of shortest \csws},
	domain=5:160,
	legend pos= north east
	]
	\addplot+[mark=.]
	plot coordinates {
		(155, 27.5)
		(158,25.37)
		(159,23.82)

	};
	\end{axis}
	\end{tikzpicture}
	\caption{Careful synchronization versus density for $80$ states}\label{carful-density}
\end{figure}

\subsection{Series 4: Influence of density}
In this series, we fixed two parameters $n$ and $\rho\le 2n-1$. For pairs $(n,\rho)$ such that $\rho\ge n+1$, we generated a sample of random binary PFAs with $n$ states, density $\rho$, and an everywhere defined letter as described in Subsect.~\ref{subsect:generating}. Then we computed the average length of shortest \csws\ for PFAs in this sample, having used the same procedure as above, that is, the pre-selection based on Lemma~\ref{lem:cyclic} followed by binary search as described in Subsect.~\ref{subsect:length}. Similarly, for pairs $(n,\rho)$ with $\rho\ge 2$, we prepared a sample of random binary PFAs with $n$ states and density $\rho$, and then we computed the average length of shortest \esws\ for these PFAs. Dealing with shortest \esws\ was slightly more involved. The complication was due to the fact that, in the absence of an everywhere defined letter, a PFA having an \esw\ of some length may have no \esw\ of any larger length. In fact, such situations occur quite often for PFAs of low density. Due to this subtlety, binary search could not be used, and therefore, we were forced to check, for each PFA $\mA$ in our sample, the SAT instances that encoded the ESW instances $(\mA,1)$, $(\mA,2)$, $(\mA,3)$, etc.

Our experiments showed that the average length of the shortest exactly synchronizing word increased as the density increased. This strongly contrasts the case of careful synchronization where the results were opposite: the more the density was, the less was the average length of shortest carefully synchronizing word. Figures~\ref{exact-density} and~\ref{carful-density} illustrate these observations.

When an automaton becomes complete, its carefully and exactly synchronizing words become nothing but classical synchronizing words of the complete case. Therefore, it is natural to expect that, when $\rho$ approaches $2n$, the average lengths of both carefully and exactly synchronizing words for synchronizing binary PFAs with $n$ states tend to the average length of synchronizing words for synchronizing binary CFAs with $n$ states. The latter length has been evaluated by Kisielewicz, Kowalski, and Szyku\l{}a in~\cite{KKS15} as a result of a series of massive experiments. Namely, the average length of synchronizing word for synchronizing binary CFAs with $n$ states is approximately equal to $2.5\sqrt{n-5}$. If one looks at the graphs in Fig.~\ref{exact-density} and~\ref{carful-density}, one may observe that they match the expectation above. Indeed, the expression $2.5\sqrt{n-5}$ gives 12.5 for $n=30$ and approximately 21.65 for  $n=80$. Extrapolating the graphs in Fig.~\ref{exact-density} and~\ref{carful-density} to the right, one gets very close values for the ordinates that would correspond to $\rho=60$ and respectively $\rho=160$.

The same behaviour was observed in our experiments with PFAs of other sizes.

\section{Benchmarks and slowly synchronizing automata}
Besides experimenting with random PFAs, we have tested our approach on certain provably `slowly synchronizing' automata, that is, the ones with the minimum length of \csws\ close to the state number squared.

We restrict ourselves to almost complete PFAs in the sense of Subsect.~\ref{subsect:probability}; recall that these are binary PFAs with only one undefined transitions. De Bondt, Don, and Zantema~\cite[Theorem~17]{dBDZ1} have proved that for any sufficiently large $n$ divisible by 10, there exists an almost complete PFA with $n$ states whose shortest \csw\ length is $\Omega(2^{\frac{n}5})$. This remarkable result has been obtained by a series of non-trivial constructions, built one on the top of others, so that it is very difficult to estimate the constant behind the $\Omega$-notation, to say nothing of exhibiting any such PFA in an explicit form. Therefore we could not test our method on these PFAs.

Fortunately, the same paper \cite{dBDZ1} provides also an explicit series of slowly synchronizing almost complete PFAs. For each $n\ge 3$, let $\mathrsfs{P}_n$ stand for the PFA with the state set $\{1,2,\dots,n\}$, on which the input letters $a$ and $b$ act as follows:
\[
q.a:=
\begin{cases}
q+1  &\text{if }  q=1,n, \\
q  &\text{if }  q=2,\dots,n-1;
\end{cases}
\qquad
q.b:=
\begin{cases}
q+1  & \text{if } q=1,\dots,n-1,\\
\text{undefined} &\text{if } q=n.
\end{cases}
\]
The automaton $\mP_4$ is the one we used as an example in Section~\ref{sec:intro}; see Fig.~\ref{fig:example} there. The automaton $\mP_n$ with $n\ge 4$ is shown in Fig.~\ref{Pn}.

Recall that the classic sequence $\mathrm{fib}(m)$ of the Fibonacci numbers is defined by the recurrence $\mathrm{fib}(m)= \mathrm{fib}(m-1) + \mathrm{fib}(m-2)$ for $m\ge 2$, together with the initial condition $\mathrm{fib}(0)=0$,  $\mathrm{fib}(1)=1$. The following result is stated in~\cite{dBDZ1} without proof:

\begin{figure}[t]
	\begin{center}
\unitlength 1mm
\begin{picture}(110,38)(0,-10)
\thinlines \node(A)(15,-5){$n{-}1$}
\node(B)(30,15){$n$} \node(C)(55,22){1}
\node(D)(80,15){2} \node(E)(95,-5){3}
\drawloop[loopangle=150](A){$a$}
\drawloop[ELpos=60,loopangle=60](D){$a$}
\drawloop[loopangle=30](E){$a$} \drawedge(A,B){$b$}
\drawedge(B,C){$a$} \drawedge[curvedepth=3](C,D){$a$}
\drawedge[curvedepth=-3](C,D){$b$} \drawedge(D,E){$b$}
\put(13,-12){\dots} \put(92,-12){\dots}
\end{picture}
\end{center}
	\caption{The automaton $\mathrsfs{P}_n$}\label{Pn}
\end{figure}

\begin{pro}\label{debondt-thm}
For $n\ge 3$, let $m$ be a unique integer that satisfies the double inequality $\mathrm{fib}(m-1)<n-2\leq \mathrm{fib}(m)$. The shortest carefully synchronizing word for the automaton $\mathrsfs{P}_n$ has length
$n^2+mn-5n-\mathrm{fib}(m+1)-2m+8$.
\end{pro}

We applied our algorithm to the automata $\mathrsfs{P}_n$ with $n=4,5,\dots,12$, and for each of them, our result matched the value predicted in Proposition~\ref{debondt-thm}. The time consumed ranged from 0.301 sec for $n=4$ to 14164 sec for $n=12$. Observe that in the latter case the shortest \csw\ has length 141 so that the `honest' binary search started with $(\mathrsfs{P}_{12},1)$ required 16 calls of MiniSat, namely, for the encodings of $(\mathrsfs{P}_{12},\ell)$ with $\ell=1,2,4,8,16,32,64,128,256,192,160,144$, $136,140,142,141$. (Of course, if one just wants to confirm (or to disprove) a theoretical prediction $\ell$ for the minimum length of \csws\ for a given PFA $\mA$, two calls of a SAT solver suffice---on the encodings of the CSW instances $(\mA,\ell)$ and $(\mA,\ell-1)$.)

\begin{figure}[b]
	\begin{center}
\unitlength 1mm
\begin{picture}(110,38)(0,-10)
\thinlines \node(A)(15,-5){$n{-}1$}
\node(B)(30,15){$n$} \node(C)(55,22){1}
\node(D)(80,15){2} \node(E)(95,-5){3}
\drawloop[loopangle=150](A){$a$}
\drawloop[ELpos=60,loopangle=60](D){$a$}
\drawloop[loopangle=30](E){$a$} \drawedge(A,B){$b$}
\drawedge(B,C){$a$} \drawedge(C,D){$a$}
\drawedge[ELside=r](B,D){$b$} \drawedge(D,E){$b$}
\put(13,-12){\dots} \put(92,-12){\dots}
\end{picture}
\end{center}
	\caption{The automaton $\mathrsfs{P}'_n$}\label{P'n}
\end{figure}
Observe that the series $\mathrsfs{P}_n$ is closely related to a series of slowly synchronizing CFAs introduced and analyzed in~\cite{AVG}; we mean the series denoted $\mathrsfs{E}_n$ in ~\cite{AVG}. Namely, $\mP_n$ and $\mE_n$ differ only in the action of $b$ at the state $n$: in $\mP_n$ this action is undefined while in $\mE_n$ one has the transition $n\stackrel{b}{\to}2$. Removing from the automaton $\mE_n$ the transition $1\stackrel{b}{\to}2$, one gets yet another series of almost complete PFAs which we denote by $\mP'_n$; see Fig.~\ref{P'n}. It turns out that the automata $\mP'_n$ also have relatively long \csws; we will derive an explicit formula for the length of the shortest \csw\ for $\mP'_n$ a little bit later.

In our experiments, whenever we encountered PFAs that had the minimum length of \csws\ close to the square of the number of states and shared some pattern, we tried to generalize these automata in order to get infinite series. Then we attempted to prove that all PFAs in these series were slowly synchronizing. We present here two of the infinite series that we found this way.

\begin{figure}[t]
	\begin{center}
\unitlength 1mm
\begin{picture}(110,38)(0,-10)
\thinlines \node(A)(15,-5){$n{-}2$}
\node(B)(30,15){$n{-}1$} \node(C)(55,22){$n$}
\node(D)(80,15){{2}} \node(E)(95,-5){{3}}
\node(F)(55,4){1}
\drawedge[ELside=r](F,D){$a$}
\drawloop[loopangle=150](A){$a$}
\drawloop[ELpos=40,loopangle=120](B){$a$}
\drawloop[ELpos=60,loopangle=60](D){$a$}
\drawloop[loopangle=30](E){$a$} \drawedge(A,B){$b$}
\drawedge(B,C){$b$} \drawedge[curvedepth=3](C,D){$a$}
\drawedge[ELpos=40,ELside=r,curvedepth=-3](C,D){$b$} \drawedge(D,E){$b$}
\put(13,-12){\dots} \put(92,-12){\dots}
\end{picture}
\end{center}
	\caption{The automaton $\mH'_n$}\label{mH'n}
\end{figure}
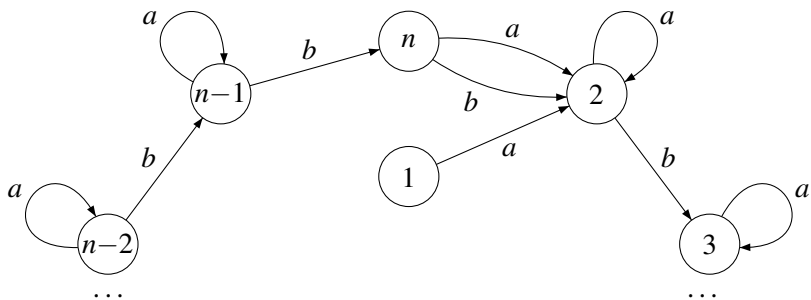

For each $n>4$, let $\mH'_n$ be the PFA with the state set $\{1,2,\dots,n\}$ on which the input letters $a$ and $b$ act as follows:
	\[
	q.a:=\begin{cases}
	2 & \text{if } q=1,2,n,\\
	q & \text{otherwise};
	\end{cases}
	\qquad
	q.b:=\begin{cases}
	\text{undefined} &  \text{if } q=1,\\
	q+1 & \text{if } 1<q<n,\\
	2  &  \text{if } q = n.
	\end{cases}
	\]
The automaton $\mH'_n$  is shown in Fig.~\ref{mH'n}. The reader acquainted with the theory of complete \sa\ immediately recognizes that the subautomaton induced by the action of $a$ and $b$ on the set $\{2,\dots,n\}$ is exactly the $(n-1)$-state automaton $\mC_{n-1}$ from the famous series discovered by \v{C}ern\'{y}~\cite{Cerny:1964} in~1964. Clearly, if a PFA $\mA$ has a subautomaton $\mB$, then every \csw\ for $\mA$ (if exists) also serves as a \csw\ for $\mB$. Hence, every \csw\ for $\mH'_n$ (if exists) must be a \sw\ for the complete subautomaton $\mC_{n-1}$. It follows from~\cite[Lemma~1]{Cerny:1964}, see also~\cite[Theorem~3]{AVG} for an easy alternative proof, that the shortest \sw\ for $\mC_{n-1}$ is the word $w:=(ab^{n-2})^{n-3}a$ of length $(n-2)^2$ which brings every state of the subautomaton to the state 2. Hence no \csw\ for $\mH'_n$ can be shorter than $w$. On the other hand, one can readily compute that $1.w=2$ as well, whence $w$ is a \csw\ for the whole automaton $\mH'_n$. We have thus established

\begin{pro}
	\label{prop:first series}
		The automaton $\mH'_n$ is carefully synchronizing and the minimum length of \csws\ for $\mH'_n$ is equal to $(n-2)^2$.
\end{pro}

Now we can return to the series $\mP'_n$ defined above. Using techniques developed in~\cite[Section~4]{AVG} for studying slowly synchronizing CFAs, we deduce the following 
\begin{cor}
\label{cor:pnprime}
The automaton $\mP'_n$ is carefully synchronizing and the minimum length of \csws\ for $\mP'_n$ is equal to $n^2-3n+2$.
\end{cor}

\begin{proof}
It is easy to verify that $(a^2b^{n-2})^{n-3}a^2$ is a \csw\ for $\mathrsfs{P}'_n$. The length of this word is equal to $n(n-3)+2=n^2-3n+2$.

Now let $w$ be a \csw\ of minimum length for $\mathrsfs{P}'_n$. Notice that in $\mathrsfs{P}'_n$, we have $q.bab=q.b^2$ for each state $q$ at which the word $bab$ is defined, that is, for each $q\ne n-1$. Besides, the words $a^3$ and $a^2$ act in $\mathrsfs{P}'_n$ in the same way. Therefore neither $bab$ nor $a^3$ can occur in the word $w$ as a factor---otherwise substituting $bab$ by $b^2$ or $a^3$ by $a^2$, one could have transformed $w$ to a shorter word that remains carefully synchronizing, a contradiction. Further, $w$ must start with $a$ since only this letter is everywhere defined but cannot start with $ab$ because $ab$ is undefined at the state $n$. Finally, let $x$ stand for the last letter of $w$ so that $w=w'x$ for some $w'\in\{a,b\}^*$. Then the minimality of $w$ implies that the image of $\{1,2,\dots,n\}$ under the action of $w'$ is equal to $\{1,2\}$ and $x=a$. The set $\{1,2\}$ is not contained in the image of the letter $b$, whence $w'$ cannot end with $b$. Thus, we conclude that $w=a^2b^{i_1}a^2b^{i_2}\cdots b^{i_k}a^2$ for some $i_1,i_2,\dots,i_k\ge 1$.

Let $c=a^2$, then the word $w$ can be rewritten into a word $v$ over the alphabet $\{b,c\}$. The actions of $b$ and $c$ on the set $\{1,2,\dots,n\}$ define an automaton shown in Fig.~\ref{fig:induced}.
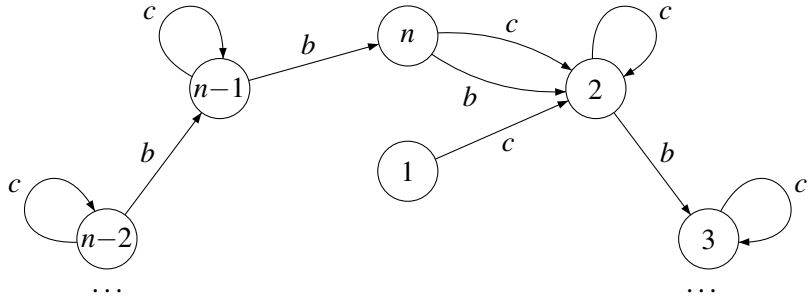
\begin{figure}[b]
	\begin{center}
\unitlength 1mm
\begin{picture}(110,38)(0,-10)
\thinlines \node(A)(15,-5){$n{-}2$}
\node(B)(30,15){$n{-}1$} \node(C)(55,22){$n$}
\node(D)(80,15){{2}} \node(E)(95,-5){{3}}
\node(F)(55,4){1}
\drawedge[ELside=r](F,D){$c$}
\drawloop[loopangle=150](A){$c$}
\drawloop[ELpos=40,loopangle=120](B){$c$}
\drawloop[ELpos=60,loopangle=60](D){$c$}
\drawloop[loopangle=30](E){$c$} \drawedge(A,B){$b$}
\drawedge(B,C){$b$} \drawedge[curvedepth=3](C,D){$c$}
\drawedge[ELpos=40,ELside=r,curvedepth=-3](C,D){$b$} \drawedge(D,E){$b$}
\put(13,-12){\dots} \put(92,-12){\dots}
\end{picture}
\end{center}
	\caption{The automaton defined by the actions of the words $b$ and $c=a^2$ in $\mP'_n$}\label{fig:induced}
\end{figure}
Since the words $w$ and $v$ act on $\{1,2,\dots,n\}$ in the same way, $v$ is a \csw\ for the latter automaton, which, obviously, is isomorphic to $\mH'_n$. By Proposition~\ref{prop:first series} the length of $v$ as a word over $\{b,c\}$ is at least $(n-2)^2$ and $v$ contains at least $n-2$ occurrences of $c$. Since every occurrence of $c$ in $v$ corresponds to an occurrence of the factor $a^2$ in $w$, we conclude that the length of word $w$ is not less than $(n-2)^2+(n-2)=n^2-3n+2$.
\end{proof}

For each $n>4$, let $\mH''_n$ be the PFA with the state set $\{0,1,\dots,n-1\}$ on which the input letters $a$ and $b$ act as follows:
	\[
	q.a:=\begin{cases}
	q+1 & \text{if } q\le n-2,\\
	1 &  \text{if } q=n-1;
	\end{cases}
	\qquad
	q.b:=
	\begin{cases}
	\text{undefined} &  \text{if } q = 0,\\
	q+1\!\!\!\pmod{n} & \text{if }  q\ge 1.\\
	\end{cases}
	\]
The automaton $\mH''_n$  is shown in Fig.~\ref{mH''n}. We observe that the automata $\mH''_n$ are closely related to the so-called Wielandt automata $\mW_n$ which play a distinguished role in the theory of complete \sa; see \cite[Theorem~2]{AVG}. Namely, $\mW_n$ is just $\mH''_n$ with the transition $0\stackrel{b}{\to}1$ added.
\begin{figure}[h]
\begin{center}
\unitlength 1mm
\begin{picture}(110,38)(0,-10)
\thinlines \node(A)(15,-5){$n{-}2$}
\node(B)(30,15){$n{-}1$} \node(C)(55,22){{0}}
\node(D)(80,15){{1}} \node(E)(95,-5){{2}}
\drawedge[ELside=r](B,D){$a$}
\drawedge[curvedepth=3](A,B){$a$}
\drawedge[ELside=r,curvedepth=-3](A,B){$b$}
\drawedge(B,C){$b$} \drawedge[curvedepth=3](D,E){$a$}
\drawedge[ELside=r,curvedepth=-3](D,E){$b$} \drawedge(C,D){$a$}
\put(13,-12){\dots} \put(92,-12){\dots}
\end{picture}
\end{center}
\caption{The automaton $\mH''_n$}\label{mH''n}
\end{figure}
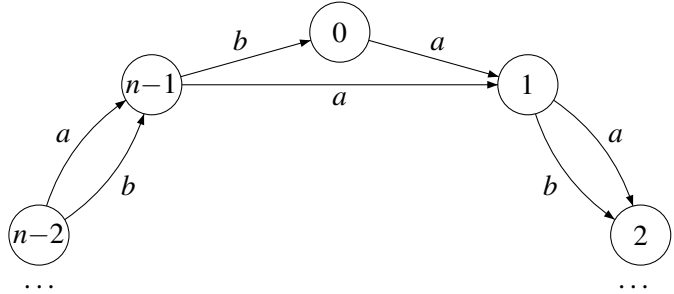

\begin{pro}
	\label{prop:second series}
	The automaton $\mH''_n$ is carefully synchronizing and the minimum length of \csws\ for $\mH'_n$ is equal to $n^2-3n+3$.
\end{pro}
\begin{proof}
Here we also use a suitable adaptation of arguments from~\cite[Section~4]{AVG}.

Suppose that $\mH''_n$ is carefully synchronizing and let $w$ be its carefully synchronizing word of minimum length. Then $w$ must bring the automaton to the state $1$; otherwise, removing from $w$ its last letter would yield a shorter \csw. Since the letter $a$ is everywhere defined, for every positive integer $i$, the word $a^{i}w$ also brings $\mH''_n$ to the state $1$. In particular, $1.a^{i}w=1$, that is, $a^iw$ labels a cycle in the underlying digraph of $\mH''_n$. Therefore, for every $\ell\ge |w|$, there is a cycle of length $\ell$ in $\mH''_n$. The underlying digraph of $\mH''_n$ has simple cycles only of two lengths: $n$ and $n-1$. Each cycle of the digraph must consist of simple cycles of these two lengths, whence each number $\ell\ge|w|$ must be expressible as a non-negative integer combination of $n$ and $n-1$. Here we invoke the following well-known and
elementary result from number theory:

\begin{lem}[{\mdseries\cite[Theorem 2.1.1]{Alfons2005}}]
\label{lem:sylvester} If $k_1,k_2$ are relatively prime positive integers,
then $k_1k_2-k_1-k_2$ is the largest integer that is not expressible as a
non-negative integer combination of $k_1$ and $k_2$.
\end{lem}

Lemma~\ref{lem:sylvester} implies that $|w|>n(n-1)-n-(n-1)=n^{2}-3n+1$. Suppose that $|w|=n^{2}-3n+2$. Since $0.w=1$, there should be a path of this length the state 0 to the state  1. The only letter defined at 0 is the letter $a$, whence $w=av$ for some $v$. Since $0.a=1$, we have $1.v=1$ so that the word $v$ labels a cycle in $\mH''_n$. However, the length of $v$ is $n^{2}-3n+1=n(n-1)-n-(n-1)$ and Lemma~\ref{lem:sylvester}. no cycles of this length may exist in the digraph of $\mH''_n$, a contradiction. Hence, $|w|\ge n^{2}-3n+3$.

On the other hand, it can be readily verified that the word $(aba^{n-2})^{n-3}aba$ of length $n(n-3)+3=n^{2}-3n+3$ carefully synchronizes the automaton $\mH''_n$. Hence $\mH''_n$ is carefully synchronizing, and $n^2-3n+3$ is the minimum length of its \csws.
\end{proof}

From the viewpoint of our studies, the series $\mH'_n$ and $\mH''_n$ are of interest as they exhibit two extremes with respect to amenability of careful synchronization to the SAT-solver approach. The series $\mH'_n$ is turned to be a hard nut to crack for our algorithm: the maximum $n$ for which the algorithm was able to find a \csw\ of minimum length is 13, and computing this word (of length 121) took almost 4 hours. In contrast, automata in the series $\mH''_n$ turn out to be quite amenable: for instance, our algorithm found a \csw\ of length 343 for $\mH''_{20}$ in 13.38 sec. We have analyzed the algorithm built in MiniSat in order to find an explanation for such a strong contrast. Our conclusion is that the superior amenability of $\mH''_n$ is due to many parallel transitions in this automaton. Whenever a binary automaton $\mA=\langle Q,\{a,b\}\rangle$ has two parallel transitions $q_j\stackrel{a}{\to}q_k$ and $q_j\stackrel{b}{\to}q_k$, our encoding of the instance $(\mathrsfs{A},\ell)$ of CSW involves the clauses $\neg y_{j,t-1}\vee\neg x_t\vee y_{k,t}$ and $\neg y_{j,t-1}\vee x_t\vee y_{k,t}$ for each $t=1,2,\dots,\ell$, see Remark~4 at the end of Section~\ref{sec:SAT}. Clearly, this pair of clauses is equisatisfiable with the single clause $\neg y_{j,t-1}\vee y_{k,t}$, and the algorithm of MiniSat seems to make good use of such simplifications of clause systems.

\section{Comparison with the partial power automaton method}
\label{sec:power}
We made a comparison between our approach and the only method for computing \csws\ of minimum length that we had found in the literature, namely, the method based on partial power automata; see~\cite[p.~295]{Martyugin14}. Given a PFA $\mA=\langle Q,\Sigma\rangle$, its \emph{partial power automaton} $\mathcal{P}(\mA)$ has the non-empty subsets of $Q$ as the states, the same input alphabet $\Sigma$, and the transition function defined as follows: for each $a\in\Sigma$ and each $P\subseteq Q$,
\[
P{\dt}a:=\begin{cases}
\{q{\dt}a\mid q\in P\} & \text{provided $q{\dt}a$ is defined for all $q\in P$}, \\
\text{undefined} & \text{otherwise}.
\end{cases}
\]
It is easy to see that  $w\in\Sigma^*$ is a \csw\ of minimum length for $\mA$ if and only if $w$ labels a minimum length path in $\mathcal{P}(\mA)$ starting at $Q$ and ending at a singleton. Such a path can be found by breadth-first search in the underlying digraph of $\mathcal{P}(\mA)$.

We implemented the above method and ran it on our samples of random PFAs. The results of the comparison are presented in Fig.~\ref{power-sat}. In this experiment we had to restrict to PFAs with at most 16 states since beyond this number of states, our implementation of the method based on partial power automata could not complete the computation due to memory restrictions (recall that we used rather modest computational resources). However, we think that the exhibited data suffice to demonstrate that the SAT-solver approach performs by far better.

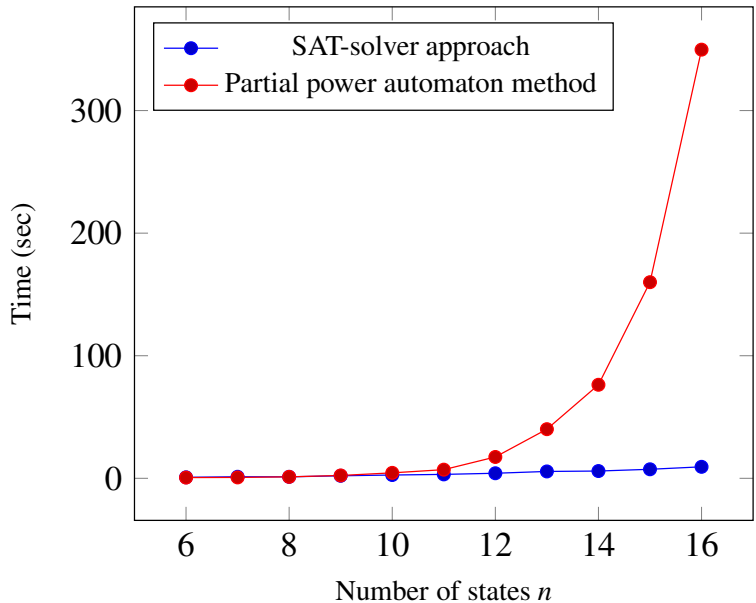
\begin{figure}[htb]
	\centering
	\begin{tikzpicture}
[scale=1.2]
	\begin{axis}[
	xlabel= {\footnotesize{Number of states $n$}},
	ylabel={\footnotesize{Time (sec)}},
	legend pos=north west
	]
	\addplot+[mark=*]
	plot coordinates {
		(6,0.82)
		(7,1.26)
		(8,1.2)
		(9,2.059)
		(10,2.68)
		(11,3.2)
		(12,4.12)
		(13,5.6)
		(14,5.93)
		(15,7.36)
		(16,9.37)
		
	};
	\addlegendentry{\footnotesize{SAT-solver approach}}
	
	\addplot+[mark=*]
	plot coordinates {
		(6,0.59)
		(7,0.72)
		(8,1.2)
		(9,2.32)
		(10,4.45)
		(11,7.06)
		(12,17.41)
		(13,40.08)
		(14,76.3)
		(15,160)
		(16,349.7)

	};
	\addlegendentry{\footnotesize{Partial power automaton method}}
	\end{axis}
	\end{tikzpicture}
	\caption{Comparison between the partial power automaton method and the SAT-solver approach}\label{power-sat}
\end{figure}

\section{Conclusion and future work}
\label{sec:final}

We have presented an attempt to approach the problem of computing a \csw\ of minimum length for a given PFA via the SAT-solver method. For this, we have developed a new encoding, which, in comparison with encodings used in our earlier papers~\cite{ShVo18,Sh18}, requires a more sophisticated proof but leads to more economic SAT instances. We have implemented and tested several algorithms based on this encoding. It turns out that our implementations work reasonably well even when a very basic SAT solver (MiniSat) and very modest computational resources (an ordinary laptop) have been employed. In order to expand the range of our future experiments, we plan to use more advanced SAT solvers. Using more powerful computers constitutes another obvious direction for improvements. Clearly, the approach is amenable to parallelization since computations needed for different automata are completely independent so that one can process in parallel as many automata as many processors are available. Still, we think that the present results, obtained without any advanced tools, do provide some evidence for our approach to be feasible in principle.

We have reported a number of experimental results. For a part for phenomena observed in the experiments, we have provided theoretical explanations but many of our observations still wait for a theoretical analysis.

At the moment, we work on designing a few new experiments based on the encoding of the present paper. In particular, we plan to investigate the so-called $D_3$-synchronization of nondeterministic automata, combining the methods of this paper with a splitting transformation described in \cite[Lemma~8.3.8]{Ito} or \cite[Section~2]{DonZantema17}. (The transformation converts any nondeterministic automaton $\mA$ into a PFA $\mA'$ over a larger alphabet such that $\mA$ is $D_3$-syn\-chro\-nizing if and only if $\mA'$ is carefully synchronizing and the minimum length of $D_3$-syn\-chro\-nizing words for $\mA$ is the same as the minimum length of carefully synchronizing words for $\mA'$.) It appears to be interesting to compare this approach with our earlier results on $D_3$-synchronization~\cite{ShVo18} based on a direct encoding of nondeterministic automata.

\section*{Acknowledgements}
We are grateful to the reviewers of the conference version~\cite{Motor} of this paper for a number of valuable remarks and suggestions.

\end{document}